\documentclass{article}
\usepackage[utf8]{inputenc}
\usepackage{graphicx}
\usepackage{xspace}
\usepackage{hyperref}
\usepackage{amsthm,amsfonts,amssymb,amsmath,array}
\usepackage{fullpage}

\usepackage{thmtools,thm-restate}
\usepackage{cleveref}
\usepackage{amsmath}
\DeclareMathOperator{\dist}{d}
\usepackage{mathtools}
\graphicspath{{./figures/}}

\newtheorem{theorem}{Theorem}
\newtheorem{lemma}[theorem]{Lemma}

\title{Pattern Formation for Fat Robots with Memory}
\author{Rusul J. Alsaedi, Joachim Gudmundsson, André van Renssen}
\date{}

\begin{document}

\maketitle

\begin{abstract}
Given a set of $n\geq 1$ autonomous, anonymous, indistinguishable, silent, and possibly disoriented mobile unit disk (i.e., fat) robots operating following Look-Compute-Move cycles in the Euclidean plane, we consider the Pattern Formation problem: from arbitrary starting positions, the robots must reposition themselves to form a given target pattern. This problem arises under obstructed visibility, where a robot cannot see another robot if there is a third robot on the straight line segment between the two robots. We assume that a robot's movement cannot be interrupted by an adversary and that robots have a small $O(1)$-sized memory that they can use to store information, but that cannot be communicated to the other robots. To solve this problem, we present an algorithm that works in three steps. First it establishes mutual visibility, then it elects one robot to be the leader, and finally it forms the required pattern. The whole algorithm runs in $O(n) + O(q \log n)$ rounds with probability at least $1 - n^{-q}$. The algorithms are collision-free and do not require the knowledge of the number of robots.
\end{abstract}

\textbf{Keywords}
\label{Keywords}
{Pattern formation \and Fat robots \and Obstructed visibility \and Collision avoidance}.

\section{Introduction}
This paper considers a system of unit disk robots in the two-dimensional Euclidean plane and studies the problem of forming a pattern using these robots. Consider a set of $n\geq 1$ mobile robots working under the classical oblivious robots model \cite{Flocchini2012} which are autonomous (no external control), anonymous (no unique identifiers), indistinguishable (no external markers), with some history (limited memory of the past), silent (no means of direct communication), and possibly disoriented (no agreement on their coordination systems). All robots execute the same algorithm. They operate following Look-Compute-Move (LCM) cycles \cite{DAS2016171} and work toward achieving a common goal, i.e., when a robot becomes active, it uses its vision to get a snapshot of its surroundings (Look), computes a destination point based on the snapshot (Compute), and finally moves towards the computed destination (Move).

There are three types of schedulers for the activation of the robots in the plane: fully synchronous,  semi-synchronous, and asynchronous. In the fully synchronous model, time is discrete and at each time instant $t$ all robots in the plane are activated and perform their LCM operations instantaneously, ending at time $t+1$. The semi-synchronous is similar, except that a subset of robots is activated at each time. Therefore, we use \emph{round} $t$ in fully synchronous and semi-synchronous instead of time $t$. In the asynchronous setting, each robot acts independently from the others, and the duration of the LCM operations of each robot is finite but unpredictable, i.e., there is no common notion of time.

This classical robot model has a long history of research and has many applications including coverage, exploration, intruder detection, data delivery, and symmetry breaking \cite{cieliebak2012distributed}. Unfortunately, most of the previous work considered the robots in this model to be dimensionless point robots which do not occupy any space.

The classical model also makes an important assumption of unobstructed visibility: any three collinear (point) robots are mutually visible to each other. However, this assumption may not reflect reality, since robots are not dimensionless points and hence they may block the view of other collinear robots. Therefore, the capabilities of robots under obstructed visibility has been the subject of recent research \cite{Agathangelou2013,Bolla2012,Cohen2008,Cord-Landwehr2011,Czyzowicz2009,di2017mutual,Luna2014,Luna2014b,Dutta2012,Flocchini2015,Chaudhuri15,Sharma2015b,Sharma2015,Sharma2015c,Vaidyanathan2015}. Under obstructed visibility, robot $r_i$ can see robot $r_j$ if and only if there is no third robot on the straight line segment connecting $r_i$ and $r_j$.

Recent research under obstructed visibility focuses on the Pattern Formation problem: Starting from arbitrary (distinct) positions in the plane, determine a schedule to reposition the robots such that they form a given target pattern without any two robots colliding during the execution of the algorithm \cite{bose2020arbitrary,bose2021arbitrary,vaidyanathan2022fast}.
We say that two robots collide if at any point in time they share the same position. The target pattern is allowed to be scaled, rotated, translated, and reflected in order to be successfully built by the robots.

As far as we are aware, there currently exist no algorithms to solve the Pattern Formation problem in the classical model for fat robots with limited memory. Our algorithm works in three phases. First it establishes mutual visibility, then it elects one robot to be the leader, and finally it forms the required pattern. During the execution, the robots do not need to agree on the orientation of any axes and the robots do not initially need to know how many robots there are. The algorithm requires in $O(n) + O(q \log n)$ rounds with probability at least $1 - n^{-q}$ in the fully synchronous model under obstructed visibility and is collision-free. 

\subsection{Related Work}
A variation of the classical model described above, called the luminous robots model (or the robots with lights model), has become the focus of significant interest \cite{bose2020arbitrary,bose2021arbitrary,di2017mutual,Luna2014,Luna2014b,Peleg2005,Sharma2015b,Sharma2015,Vaidyanathan2015,vaidyanathan2022fast}. In this model, robots are equipped with an externally visible light which can assume different colors from a fixed set of colors. The lights are persistent, i.e., the color of the light is not erased at the end of the LCM cycle. Except for the assumption of the availability of lights, the robots work similarly to the classical model. This model corresponds to the classical oblivious robots model when the number of colors in the set is 1 \cite{di2017mutual,Flocchini2012}. In this model, the objective is to solve the Pattern Formation problem while minimizing the size of the color set.

Most of the existing work considered dimensionless point robots to solve the Pattern Formation problem \cite{bose2022distributed,bose2020pattern,fujinaga2015pattern,chaudhuri2014pattern,ghosh2022move,kundu2022arbitrary,pattanayak2020distributed}, and most of them considered the above lights model while solving the problem \cite{bose2021arbitrary,vaidyanathan2022fast}. However, the techniques used in those papers do not generalize to fat robots and cannot be used to solve the problem in the classical model. Most closely related to our work is the paper by Bose {\it et al.} \cite{bose2020arbitrary}, where they studied the Pattern Formation problem for fat robots under the asynchronous scheduler. In their solution, using the robots with lights model, all the robots need to agree on one axis. They solved the problem using ten colors. The required number of colors was recently improved by Alsaedi {\it et al.}~\cite{alsaedi2023pattern}, who gave algorithms that required 7 or 8 colors, depending on whether the target pattern was allowed to be scaled or not. 

Other recent work \cite{bose2021arbitrary,vaidyanathan2022fast} solved the Pattern Formation problem in the robots with lights model assuming that the robots are dimensionless points. These papers provided techniques to overcome obstructed visibility, but did not take the physical extents of the robots into account. Unfortunately, the techniques developed for point robots cannot be applied directly to solve Pattern Formation for fat robots for this reason. 

Kundu {\it et al.} \cite{kundu2022arbitrary} studied the Pattern Formation problem for fat robots with lights on an infinite grid with one axis agreement of robots. They solved the problem using nine colors. However, they did not bound the running time of their algorithm.

Other results for the Pattern Formation problem include those by Cicerone {\it et al.} \cite{cicerone2021solving}, who considered point robots with chirality (robots have to agree on a notion of left and right), and those by Flocchini {\it et al.} \cite{flocchini2008arbitrary}, who studied the problem for point robots in the asynchronous setting, requiring that the robots agree on their environment and observe the positions of the other robots.

\section{Preliminaries}
Consider a set of $n\geq 1$ anonymous robots $R=\{r_1,r_2,\ldots,r_n\}$ operating in the Euclidean plane $\mathbb{R}^2$. The number of robots is not assumed to be known to the robots.
We assume that each robot $r_i\in R$ is a non-transparent disk with diameter 1. The center of the robot $r_i$ is denoted $c_i$, and the position of $c_i$ is also said to be the position of $r_i$. We denote by $\dist(r_i,r_j)$ the Euclidean distance between the two robots $c_i$ to $c_j$. For simplicity, we use $r_i$ to denote both the robot $r_i$ and the position of its center $c_i$. 
Each robot $r_i$ has its own coordinate system, and it knows its position with respect to this coordinate system.
Note that robots do not necessarily agree on the orientation of their coordinate systems. However, since all the robots are of unit size, they implicitly agree on the notion of unit length. 

Each robot has a camera to take a snapshot of the plane and the distance that is visible to this camera is infinite, provided that there are no obstacles blocking its view (i.e., another robot) \cite{Agathangelou2013}. 
Following the fat robot model \cite{Agathangelou2013,Czyzowicz2009},
we assume that a robot $r_i$ can see another robot $r_j$ if there exists a point on the bounding circle of $r_j$ that is visible to some point on the bounding circle of $r_i$.
Similarly, we say that a point $p$ in the plane is visible to a robot $r_i$ if there is a point $p_i$ on the bounding circle of $r_i$ such that the straight line segment $\overline{p_i p}$ does not intersect any other robot. 

Each robot is equipped with a small amount of memory. We assume that a single unit of memory suffices to store an integer of value at most $n$. In other words, a single unit of memory can store $O(\log n)$ bits. We also assume that the locations of the robots can be stored efficiently, using a constant number of units of memory, or a constant amount of memory for short. As all previous work assumes that the robots can perform computations based on the observed (locations of) robots, the only difference between this model and the existing model is that we allow a small number of things to be stored between rounds. Crucial is that the robots have no way to communicate with each other (i.e., they can send no messages and have no lights to indicate their internal state) and thus they cannot combine what is stored in their respective memories. 

At any time $t$, a robot $r_i\in R$ is either active or inactive. When active, $r_i$ performs a sequence of {\em Look-Compute-Move} (LCM) operations:
\begin{itemize}
\item {\em Look:} the robot takes a snapshot of the positions of the robots visible to it in its own coordinate system; 
\item {\em Compute:} executes its algorithm using the snapshot which returns a destination point $x\in \mathbb{R}^2$, potentially using the information stored in its memory and updating what is stored in its memory; and
\item {\em Move:} moves to the computed destination $x\in \mathbb{R}^2$ (if $x$ is different from its current position).
\end{itemize}

Each robot executes the same algorithm locally every time it is activated and a robot's movement cannot be interrupted by an adversary. Two robots $r_i$ and $r_j$ are said to {\em collide} at time $t$ if the bounding circles of $r_i$ and $r_j$ share a common point at time $t$. To avoid collisions among robots, we thus have to ensure that at all times $\dist(r_i,r_j)\geq 1$ for any robots $r_i$ and $r_j$.

The Pattern Formation problem is now defined as follows: Given any initial positions of the robots, the robots must reposition themselves to form a given target pattern without having any collisions in the process, and then terminate in this configuration. The target pattern is given as input as a list of locations, one for each robot in the pattern and allowed to be scaled, rotated, translated, and reflected. An algorithm is said to solve the Pattern Formation problem if it always achieves the target pattern from any initial configuration. The pattern should be constructed if this is indeed possible. We measure the quality of the algorithm using the amount of memory required to execute it and the number of rounds needed to solve the Pattern Formation problem.

\section{Mutual Visibility}
Our algorithm works in three phases. The first phase is the Mutual Visibility phase. The goal of this phase is to move the robots such that at the end of the phase every robot can see all other $n-1$ robots. The algorithm we present forms a convex hull with each robot positioned on a corner of the convex hull (see Figure~{\ref{fig:4.4}} for an example). We will argue later that our approach runs in $O(n)$ rounds and uses $O(1)$ memory space.

\begin{figure}[ht]
 \centering
  \includegraphics{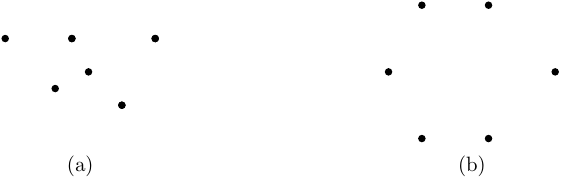}
  \caption{An example of the Mutual Visibility phase: (a) an initial configuration and (b) the end configuration of this phase. Throughout the paper the robots are shown as dots for simplicity.}
  \label{fig:4.4}
\end{figure}

\subsection{Algorithm}
On a high level, our Mutual Visibility algorithm is similar to the algorithm by Alsaedi {\it et al.} \cite{https://doi.org/10.48550/arxiv.2206.14423}. However, our algorithm works in the classical robots model (i.e., without the lights required by Alsaedi {\it et al.}). Their algorithm used two colors to distinguish (and communicate) which robots are corners of the convex hull and which robots are in the interior. Their algorithm proceeds to gradually expand the convex hull by moving the corner robots one unit away from the interior of the convex hull using the angle bisectors of consecutive robots on the hull. This ensures that the edges of the hull will become long enough to let all interior robots move through them to the outside of the convex hull to become new corners without collisions. Alsaedi {\it et al.} \cite{https://doi.org/10.48550/arxiv.2206.14423} showed that this process finishes in $O(n)$ rounds. 

We first observe that robots can easily determine whether they themselves are interior or corner robots. Corner robots are those robots on the vertices of the convex hull with interior angles less than $180^\circ$. There could initially be robots with interior angle equal to $180^\circ$, but since all the corner robots move in every round to expand the convex hull these robots will become interior robots almost immediately. 

Let us look a bit more closely at how Alsaedi {\it et al.} \cite{https://doi.org/10.48550/arxiv.2206.14423} moved the interior robots. They define \emph{eligible edges} for each robot as the set of convex hull edges of length at least $3$ for which no other robot is closer to the edge. If an interior robot has eligible edges, it picks the closest one and moves perpendicular to this edge through it to become a new corner of the convex hull. In every round, there are at most two interior robots that move through any single edge of the convex hull. By construction, the movements of the robots through different edges do not interfere each other. Therefore, there is no collision among interior robots. Since all the corner robots move in every round, the edges of the hull are long enough to move the interior robots through without any collisions. Therefore, there is no collision among all robots. Since a robot can determine whether it is interior or not, and all other computations for this process are based on the locations of the other robots it sees, this process does not require the colors and can be mimicked in our algorithm. 

Hence, it remains to argue that the robots can correctly determine when this phase has ended. Since there are no lights to signal this (all robots having the corner robots light does this in the original algorithm) and the robots do not know $n$, it is difficult for the robots to know whether the phase has ended even if the robots have a little memory. Specifically, there are cases where from a corner robot's perspective all other visible robots are in convex position. For example, in Figure~{\ref{fig:4.1}} robot $r_i$ cannot see robot $r_k$ since robot $r_j$ blocks its view to that robot. However, since all robots that are visible to $r_i$ are in convex position, it can now incorrectly conclude that it can see all robots. 

\begin{figure}[ht]
 \centering
  \includegraphics{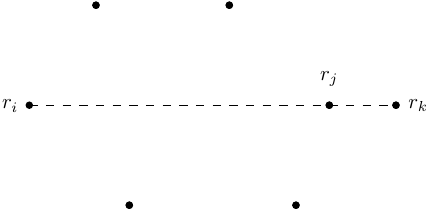}
  \caption{An example configuration where a corner robot $r_i$ believes all robots are in convex position when there is an interior robot $r_j$ left which blocks visibility to robot $r_k$.}
  \label{fig:4.1}
\end{figure}

Fortunately, the above problem can only occur for a single robot in the configuration. All other robots can see robot $r_j$ as the interior robot it is. In the next section we will prove this claim. Now, what we need to avoid is robot $r_i$ incorrectly and prematurely ending the Mutual Visibility phase for itself, as we cannot correct this easily due to the lack of communication between robots. Hence, when a corner robot believes the phase is over because it observed that all other visible robots are in convex position, it starts a counter with value $4k+2$, where $k$ is the number of visible corner robots. It stores both the value $k$ and the counter in its memory. Every round, it continues expanding the convex hull as normal (i.e., it continues to behave the same way as if not all robots are in convex position yet), but it also decrements the counter. This is to ensure that if at any point it realizes that it made a mistake, we do not have to deal with correcting the situation or computing where it should have moved due to a mistake a certain number of rounds ago. If it sees more than $k$ corner robots within these $4k+2$ rounds, it stops its current counter and concludes that it was previously incorrect. Because the robot acted as normal, it is already in the exact same location as it would have been if it had not made this mistake and thus the algorithm can proceed as normal (and if it again observes all visible robots in convex position, it will start a new counter using the new higher value of $k$). In the next section, we will argue that this process ensures that every corner robot correctly concludes that the Mutual Visibility phase has ended and that they all conclude this in the same round. As a result, the Mutual Visibility problem has been solved. All robots store this in their memory. 

\subsection{Analysis}
We proceed to prove that our algorithm solves the Mutual Visibility problem in a linear number of rounds using a constant amount of memory, while avoiding collisions between the robots.

Since our algorithm is similar to the algorithm by Alsaedi {\it et al.} \cite{https://doi.org/10.48550/arxiv.2206.14423}, the main parts of the proof can be found there. We therefore focus on the differences between the two algorithms.

We start by recalling that contrary to the original algorithm, robots cannot always determine whether they are in convex position. Specifically, if multiple robots are collinear, some might observe all others to be in convex position, while globally this is not yet the case (see Figure~{\ref{fig:4.1}}). We first argue that there can be at most one robot that incorrectly concludes that all robots are in convex position.

\begin{lemma}
\label{lemma:4.0}
If there exists an interior robot, then there exists at most one corner robot that believes all robots are in convex position.
\end{lemma}
\begin{proof}
We observe that for a robot not to be visible to a corner robot, three robots have to be collinear. This implies that all robots except the two extreme ones on this line can conclude that they are not in convex position. 

Now consider such an extreme robot. If there are one or more interior robots that block the view of one corner robot, this corner robot observes the position of at least one interior robot by using the line segments connecting each pair of corner robots (see Figure~\ref{fig:4.7}). 

\begin{figure}[ht]
 \centering
  \includegraphics[scale=0.8]{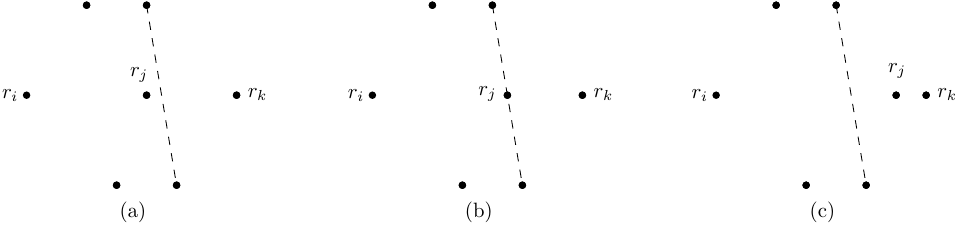}
  \caption{Robot $r_j$ blocks visibility between $r_i$ and $r_k$: (a) $r_i$ concludes that $r_j$ is an interior robot, (b) $r_i$ concludes that the robots are not yet in convex position, and (c) $r_i$ incorrectly concludes that the robots are in convex position, but $r_k$ still correctly concludes that this is not the case.}
  \label{fig:4.7}
\end{figure}

If an interior robot $r_j$ lies before a line segment between a pair of corner robots (see Figure~\ref{fig:4.7}a), the corner robot $r_i$ observes that they are not in convex position. If the interior robot $r_j$ lies on the line segment between a pair of corner robots (see Figure~\ref{fig:4.7}b), the corner robot $r_i$ observes that there are three collinear robots and thus the robots are not yet in convex position. If an interior robot lies behind the line segment between a pair of corner robots (see Figure~\ref{fig:4.7}c), corner robot $r_i$ may believe the robots are in convex position because the interior robot blocks the view to the corner robot $r_k$ that lies behind it. Hence, from $r_i$'s perspective, all robots are in convex position. However, note that from $r_k$'s perspective $r_j$ lies before this same line segment and thus $r_k$ correctly concludes that the robots are not in convex position yet. Hence, only one robot can incorrectly conclude that all robots are in convex position and the lemma follows.
\end{proof}

Next, we argue that when a robot starts a counter, within $4k+2$ rounds either it sees more corner robots or it can correctly conclude that mutual visibility has been achieved.

\begin{lemma}
\label{lemma:4.1}
If a robot $r$ starts a counter with value $4k+2$ (where $k$ is the number of currently visible corner robots), then within that number of rounds either it sees more than $k$ corner robots or mutual visibility has been achieved, and $r$ terminates this phase.
\end{lemma}
\begin{proof}
Lemma $8$ in \cite{https://doi.org/10.48550/arxiv.2206.14423} shows that it takes at most $4k$ rounds to make enough space for the interior robots to move outside the current convex hull. In addition to that, it takes two more rounds to move at least one interior robot outside the current convex hull. Therefore, if there is an interior robot inside the hull, it takes at most $4k+2$ rounds for a new corner robot to appear. Since all other corner robots remain visible, this implies the first part of the lemma. Otherwise, if there is no interior robot inside the hull, then mutual visibility has been achieved, and $r$ correctly terminates this phase.
\end{proof}

Now that we know that at some point the robots conclude that the phase has concluded, we proceed to show that all robots terminate in the same round.

\begin{lemma}
\label{lemma:4.2}
All robots terminate the Mutual Visibility phase in the same round.
\end{lemma}
\begin{proof}
Since by Lemma~\ref{lemma:4.0}, only one corner robot can incorrectly believe that all robots are in convex position, only this robot can start a counter with value $4k+2$ by Lemma~\ref{lemma:4.1}. If there is an interior robot, by Lemma \ref{lemma:4.1} it will conclude this within $4k+2$ rounds.  

Once there are no interior robots, every corner robot starts a counter with value $4n+2$ when the last corner robot joins the convex hull, i.e., in the same round. Therefore, all robots terminate the Mutual Visibility phase in the same round.
\end{proof}

It remains to analyze the time complexity of the Mutual Visibility phase.

\begin{lemma}
\label{lemma:4.3}
The Mutual Visibility phase takes $O(n)$ rounds.
\end{lemma}
\begin{proof}
This lemma follows from Lemma $8$ in \cite{https://doi.org/10.48550/arxiv.2206.14423}, which states that the original algorithm takes $O(n)$ rounds, and Lemma~\ref{lemma:4.1}, which adds an additional $4n+2$ (i.e., $O(n)$) rounds for all robots to correctly conclude that they are in convex position.
\end{proof}

\begin{lemma}
\label{lemma:4}
The Mutual Visibility phase uses $O(1)$ memory.
\end{lemma}
\begin{proof}
In the Mutual Visibility phase, every robot needs to store the termination counter and the number of visible robots at that time. Both of these are upper bounded by $O(k)$ which is at most $O(n)$ and can thus be stored in $O(\log n)$ bits, i.e., $O(1)$ memory space. The additional bits needed to store the current phase do not affect this bound. 
\end{proof}

Therefore, we conclude the following theorem.

\begin{theorem}
\label{theorem:4.4}
The Mutual Visibility phase finishes in $O(n)$ rounds without collisions in the fully synchronous setting using $O(1)$ memory.
\end{theorem}

\section{Leader Election}
The second phase of our algorithm is the Leader Election phase. The goal of the Leader Election phase is for a single robot to be elected as a leader. It is known that this problem cannot be solved with a deterministic algorithm in an anonymous distributed system \cite{attiya2004distributed}. We present an algorithm which is an adaptation of the slotted-Aloha protocol \cite{vaidyanathan2022fast}. 

\subsection{Algorithm}
After the Mutual Visibility phase, every robot sees all other $n-1$ robots and thus $n$ is known and can be stored by all robots. This implies that at this point all robots can determine whether it is possible to build the given pattern with the number of available robots: if the number of robots in the given pattern is at most $n$, then the robots proceed and otherwise, they all terminate. Hence, in the remainder, we assume that the number of robots required to build the pattern is at most $n$. 

The next phase of our algorithm aims to elect a single robot to be the leader in the final Pattern Formation phase. The leader will store that it is the leader and when a leader is elected all robots will be aware of this fact. However, but since all robots are identical, once a leader is elected, the other (non-leader) robots will not be able to store which robot is the leader. All that each of these robots can remember is that it is not the leader. 

In order to elect a leader all robots first compute the centroid $c$ of the convex hull formed by the robots. Next, they all determine the distance $d$ from the centroid to the robot farthest from the centroid. Since all robots are visible to each other, all robots will compute the same $c$ and $d$. The leader is selected among all robots on the circle centered at $c$ having radius $d$. In other words, all robots on the boundary of the circle are \emph{competing} to become the leader and all robots in the interior of the circle are \emph{non-competing}. By construction there are no robots outside the circle and there is at least one robot on the boundary of this circle, but there could be multiple and thus we need a mechanism to pick one of them. 

To facilitate this, the following process is repeated until a leader has been elected: All competing robots flip a coin where the probability of success is $1/k$, where $k$ is the number of competing robots. If a robot is successful, then the robot remains in its current position on the boundary of the circle. Otherwise, it stores its current location on the circle and moves inside the circle some distance without crossing the straight line between its two neighbours, to ensure that the convex hull remains convex and thus all robots remain visible to each other (see Figure~{\ref{fig:4.5}}). If there are $0$ or more than $1$ robot on the boundary of the circle, the robots conclude that this iteration failed to elect a leader and the robots that were competing in the previous round return to their positions on the boundary of the circle using their stored locations to repeat the algorithm at the next iteration. When there is exactly one competing robot left on the boundary of the circle, this robot becomes the leader and it stores this information. Note that all robots can observe this event, as we keep the convex hull convex and thus the robots remain mutually visible throughout this process. Hence, all robots can determine whether a leader has been elected and whether they are the leader. Every robot stores this information, along with the fact that the Leader Election phase has ended and proceeds to the next phase of the algorithm. 

\begin{figure}[ht]
 \centering
  \includegraphics[scale=0.8]{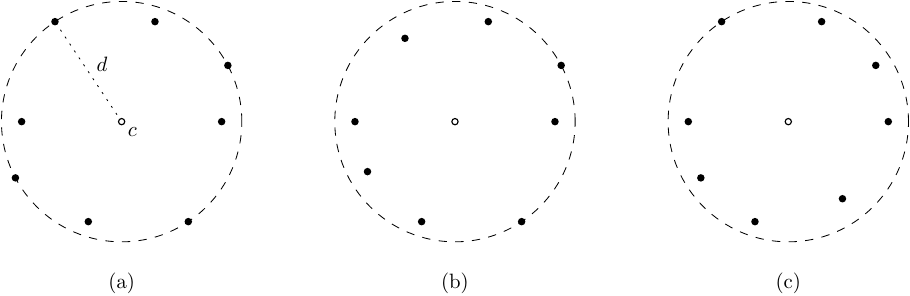}
  \caption{An example of the Leader Election phase: (a) centroid $c$ and distance $d$ are computed and used to form a circle, (b) an unsuccessful iteration where two robots remain on the boundary of the circle, and (c) a successful iteration where a leader is elected and the phase ends.}
  \label{fig:4.5}
\end{figure}

\subsection{Analysis}
We start by arguing the correctness of the Leader Election phase.

\begin{lemma}
\label{lemma:4.5}
Leader election can be solved in the memory model with $n$ robots.
\end{lemma}
\begin{proof}
The analysis is identical to that in \cite{vaidyanathan2022fast} with the lights replaced by being on the circle or not being on the circle. As this can be observed by all robots at all time, the lemma follows.
\end{proof}

Next, we consider the memory required by our algorithm. 

\newpage
\begin{lemma}
\label{lemma:4.6}
The Leader Election phase uses $O(1)$ memory.
\end{lemma}
\begin{proof}
In the Leader Election phase, every robot needs to store the circle, i.e., the center and the radius. Since both are a linear combination of the locations of the robots, these can indeed be computed and stored in $O(\log n)$ bits. In order to allow competing robots to move back to their original positions when an iteration is unsuccessful, we store their previous location, which takes $O(\log n)$ bits. Finally, every robot needs a constant number of bits to store whether it is a leader and whether the phase has ended. In total, this requires $O(\log n)$ bits of memory, i.e., $O(1)$ memory space. 
\end{proof}

Finally, since the expected running time of the Leader Election phase has already been analyzed in the asynchronous setting \cite{vaidyanathan2022fast}, and the fully synchronous setting takes at most the same number of rounds, we obtain the following result.

\begin{theorem}
\label{theorem:4.7}
For any $q > 0$, the Leader Election phase finishes in $O(q \log n)$ rounds, with probability at least $1 - n^{-q}$, without collisions in the fully synchronous setting using $O(1)$ memory.
\end{theorem}

\section{Pattern Formation}
The final phase of our algorithm is the Pattern Formation phase. The goal of this phase is to move the robots to form the given (target) pattern and thus solve the original problem. 

\subsection{Algorithm}
Since we now have a leader, this robot will facilitate the other robots moving and the other robots will not move unless instructed to do so by the leader. 
Our algorithm depends on the leader's position on the convex hull at the end of the Leader Election phase. Since the leader is positioned on the convex hull of the robots, this means that at least one of its quadrants does not contain any robots. Once the leader has determined its empty quadrant, it proceeds to build the target pattern in that quadrant. Without loss of generality, we assume that the leader's top left quadrant is empty, and thus we want to build the pattern there. We define an ordering on the robots (excluding the leader) $r_1, \dots, r_{n-1}$ and an ordering on the positions in the pattern $p_1, \dots, p_k$ ($k \leq n$). Both orderings are left to right, with ties being broken top to bottom. The high-level idea is that the leader moves robot $r_i$ to pattern position $p_i$ repeatedly until the full pattern is built. If needed, the leader itself will fill position $p_n$. 

To avoid collisions during this process, the Pattern Formation phase requires the given pattern to be scaled, as allowed by the problem. This is to ensure that after leading a robot to its position in the pattern, the leader can move away to get the next robot without colliding with previously placed robots in the process. 
Since the pattern is given as part of the input, every robot knows it and thus also knows the original size of the target pattern. The leader computes the size of the pattern after scaling by multiplying the input original size of the pattern by the scaling factor (say $5$) to get the final target pattern after scaling. Note that the scaled pattern does not need to be stored, as the leader can simply recompute it when needed. The leader stores the scaling factor, as well as the first and last locations of the pattern after scaling in order to reconstruct the placement of the robots in the scaled pattern. 

In order to precisely define the robots' movements, we need to determine some coordinates. After the Leader Election phase, let the leader be positioned at $(x,y)$. Let $y_{max}$ be the $y$-coordinate of the topmost robot on the convex hull and let $x_{min}$ be the $x$-coordinate of the leftmost robot on the convex hull. To ensure that there is no overlap between the convex hull and the scaled target pattern, the leader will build the pattern such that $p_k$ is at $(x_{min}-100,y_{max}+100)$ in the leader's coordinate system. 

Next, we describe how the leader moves the robots in more detail (see also Figure~{\ref{fig:4.6}}). Initially the leader lies on the circle of the Leader Election phase and in order to move away from that without colliding with any other robot, the leader first moves perpendicular to the circle until it reaches the $x$-coordinate $x_{min}-1$ in one round. It then moves to $(x_{min}-1,y(r_1))$, i.e., one position to the left of the leftmost topmost robot $r_1$ in the ordering. Robot $r_1$ observes that there is a robot (the leader) touching it and concludes that this robot must be the leader and that it needs to follow this robot. The problem now is that the leader cannot communicate with robot $r_1$ where it needs to go and since $p_1$ can essentially be arbitrarily far away, guiding the robot using repeated unit distance steps is not feasible. 

\begin{figure}[ht!]
 \centering
  \includegraphics[scale=0.8]{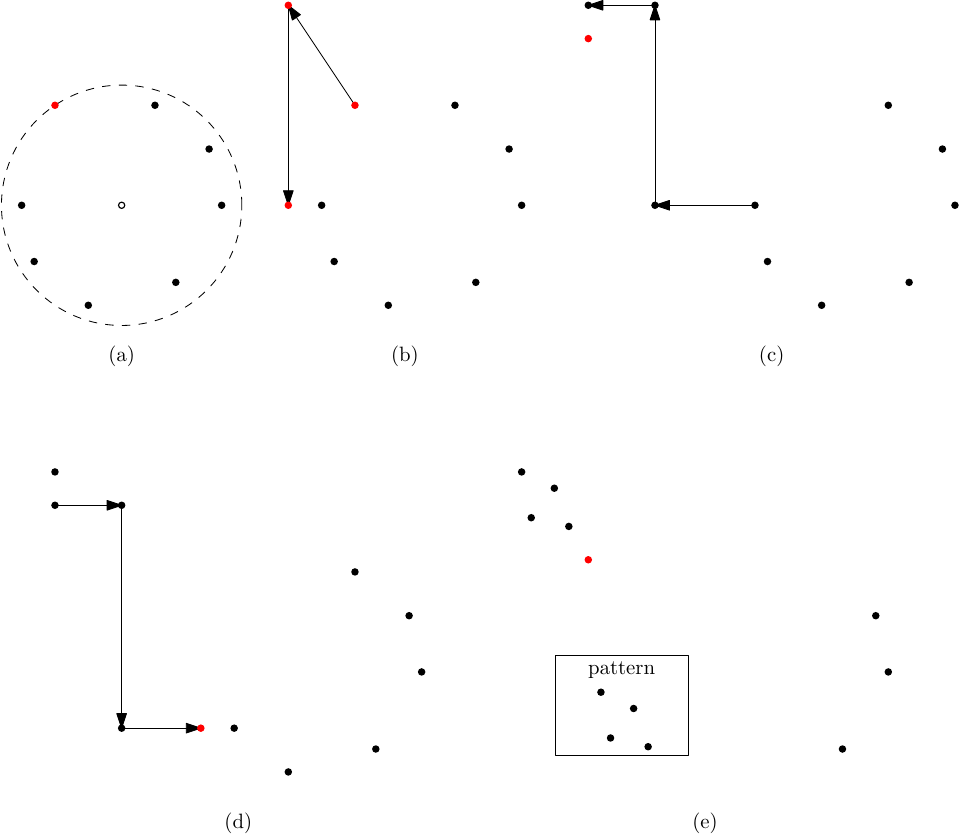}
  \caption{An example of the Pattern Formation phase (leader is colored red): (a) the situation after a leader has been elected, (b) the leader moves next to $r_1$, (c) the leader moves $r_1$ to $p_1$, (d) the leader moves next to the next robot, and (e) the pattern is built and the phase ends.}
  \label{fig:4.6}
\end{figure}

To work around this issue, the leader "instructs" the robot in two rounds. In the first, it places itself next to the robot in the direction it wants the robot to move in. Robot $r_1$ observes this and stores this direction, for example using the location of the leader. In the next round, the leader moves the intended distance in the direction of intended movement. Robot $r_1$, having stored the direction, sees how far the closest robot in that direction is and stores this distance. Once the leader has moved out of the way, robot $r_1$ is now free to move to the indicated location. Whenever the leader wants to move a robot somewhere this process is executed a constant number of times, to facilitate multiple sequential moves. 

Moving a robot $r_i$ to its position $p_i$ always starts with the leader moving one position to the left of robot $r_i$, to $(x(r_i)-1,y(r_i))$. The leader uses the above approach to guide robot $r_i$ horizontally left to $(x_{min}-90, y(r_i))$. The leader then moves to $(x_{min}-90, y(r_i)+1)$ (one position above the robot) to let robot $r_i$ store the next direction and follows the leader until it reaches $(x_{min}-90, y(p_i))$ by moving vertically up. The leader is now at $(x_{min}-91, y(p_i))$ (again one position to the left of $r_i$) to guide $r_i$ to its final position $p_i$ at $(x(p_i), y(p_i))$ by moving horizontally left. In order to ensure that $r_i$ does not follow the leader any further, the leader moves two positions down to $(x(p_i), y(p_i)-2)$. Since the target pattern is scaled, this extra movement of the leader cannot cause collisions with other robots in the target pattern, as in the original pattern robots are at least a distance of $1$ apart and thus now they are at least a distance of $5$ apart, which is more than the required $2$ units of movement plus two times the radius of a robot. Once robot $r_i$ has reached $p_i$, the leader moves horizontally right to $(x(r_{i+1})-1, y(p_i)-2)$, then vertically down to $(x(r_{i+1})-1, y(r_{i+1}))$ (the position to the left of robot $r_{i+1}$). 

The above process is repeated until a robot is placed at $p_k$ or the final robot $r_{n-1}$ is positioned at $p_{k-1}$ when $k=n$. To create a clean pattern, without the leader remaining in it, the leader moves horizontally right to $x$-coordinate $x_{min}-90$. If $k < n$, the phase now ends. Otherwise $k=n$ and the leader moves vertically to $(x_{min}-90, y(p_n))$, and finally to $p_n$ at $(x(p_n), y(p_n))$. 

\subsection{Analysis}
We now argue the correctness of the Pattern Formation phase.

\begin{lemma}
\label{lemma:4.8}
In every round of the Pattern Formation phase, only one non-leader robot moves from the convex hull to be positioned in the target pattern avoiding collisions.
\end{lemma}
\begin{proof}
In the Pattern Formation phase, the leader moves to the left of the leftmost topmost robot $r_i$ a distance one outside the convex hull. Robot $r_i$ observes this and follows the leader to its final position. Once the leader moves away from $r_i$, $r_i$ stops at its final position. Since the leader moves only one robot at a time and a robot moves only once activated by the leader, only one non-leader robot moves in each round. Furthermore, due to the order in which the robots are moved (left to right) and the first move being horizontally to the left, no collisions can occur.
\end{proof}

Next, we analyze the time complexity of the Pattern Formation phase.

\begin{lemma}
\label{lemma:4.9}
The Pattern Formation phase takes $O(n)$ rounds.
\end{lemma}
\begin{proof}
By Lemma \ref{lemma:4.8}, the leader activates and moves each robot from the convex hull to the target pattern one robot at a time. The leader moves next to a robot in a constant number of rounds and it moves a robot to its final position in three moves, so in six rounds. Therefore, one robot moves from the convex hull to the target pattern in a constant number of rounds. As a result, the total number of rounds required to move all robots from the convex hull to the target pattern is $O(n)$.
\end{proof}

Next, we prove the memory space that our algorithm needs in total for all phases.

\begin{lemma}
\label{lemma:4.10}
The Pattern Formation phase uses $O(1)$ memory.
\end{lemma}
\begin{proof}
During the Pattern Formation phase, the leader stores the first and last locations of the pattern after scaling and the scaling factor. As these are simply locations in the plane, they can be stored using $O(\log n)$ bits. Every other robot stores the direction and distance while moving to its final location in the target pattern. As these can be stored as locations of the leader, this also requires $O(\log n)$ bits. Hence, $O(1)$ memory suffices for all robots. 
\end{proof}

Together, Lemmas~\ref{lemma:4.8}, \ref{lemma:4.9}, and \ref{lemma:4.10} imply the following result. 

\begin{theorem}
\label{theorem:4.11}
The Pattern Formation phase finishes in $O(n)$ rounds without collisions in the fully synchronous setting using $O(1)$ memory.
\end{theorem}

Theorems~\ref{theorem:4.4}, \ref{theorem:4.7}, and \ref{theorem:4.11} now imply our final result. 

\begin{theorem}
\label{theorem:4.12}
Our algorithm solves the Pattern Formation problem for $n$ unit disk robots in $O(n) + O(q \log n)$ rounds with probability at least $1 - n^{-q}$ without collisions in the fully synchronous model using $O(1)$ memory space.
\end{theorem}

\section{Conclusion}
We studied the Pattern Formation problem for a system of $n$ autonomous fat robots of unit disk size in the classical model where all robots have a small persistent memory. We described an algorithm that solves the Pattern Formation problem and works under the fully synchronous model and under obstructed visibility. Our algorithm solves the Pattern Formation problem in $O(n) + O(q \log n)$ rounds with probability at least $1 - n^{-q}$. The Pattern Formation problem uses $O(1)$ memory space. As a by-product, we also solve the Mutual Visibility problem in our model.

For future work, it is interesting to extend our algorithm to the semi-synchronous and asynchronous setting. Furthermore, it is interesting to extend our algorithm to solve the Pattern Formation problem in $O(n)$ time complexity, which requires removing the Leader Election phase.

\bibliographystyle{plain}
\bibliography{references}

\end{document}